%% file: main.tex
\documentclass[sigconf]{acmart}
\usepackage{booktabs} % For formal tables
\usepackage{tikz}
\usepackage{subcaption}
\usepackage[linesnumbered,ruled]{algorithm2e}
\usepackage{enumitem}
\usepackage{multirow}
\usepackage{graphics}
\usepackage{array}
\usepackage{tabularx}
\newcolumntype{P}[1]{>{\centering\arraybackslash}p{#1}}
\newcolumntype{M}[1]{>{\centering\arraybackslash}m{#1}}
\graphicspath{ {figs/} }
\usepackage{balance}
\usepackage{flushend}
\theoremstyle{plain}
\newtheorem{theorem}{Theorem}[section]

\newtheorem{prop}[theorem]{Proposition}

\theoremstyle{definition}
\newtheorem{definition}{Definition}[section]
\begin{document}
\title{Spectral Collaborative Filtering}
\author{Lei Zheng}
\affiliation{
\institution{Department of Computer Science \\ University of Illinois at Chicago}
\state{IL}
\country{US}
}
\email{lzheng21@uic.edu}
\author{Chun-Ta Lu}
\affiliation{
\institution{Department of Computer Science \\ University of Illinois at Chicago}
\state{IL}
\country{US}
}
\email{clu29@uic.edu}
\author{Fei Jiang}
\affiliation{
\institution{Department of Computer Science \\ Peking University}
\city{Beijing}
\country{China}
}
\email{fei.jiang1989@pku.edu.cn}
\author{Jiawei Zhang}
\affiliation{
\institution{IFM Lab, Department of Computer Science \\ Florida State University}
\city{FL}
\country{US}
}
\email{jiawei@ifmlab.org}

\author{Philip S. Yu}
\affiliation{
\institution{Department of Computer Science \\ University of Illinois at Chicago}
\state{IL}
\country{US}
}
\affiliation{
\institution{Institute for Data Science \\ Tsinghua University}
\state{Beijing}
\country{China}
}
\email{psyu@uic.edu}

\begin{abstract}
Despite the popularity of Collaborative Filtering (CF), CF-based methods are haunted by the \textit{cold-start} problem, which has a significantly negative impact on users' experiences with Recommender Systems (RS). In this paper, to overcome the aforementioned drawback, we first formulate the relationships between users and items as a bipartite graph. Then, we propose a new spectral convolution operation directly performing in the \textit{spectral domain}, where not only the proximity information of a graph but also the connectivity information hidden in the graph are revealed. With the proposed spectral convolution operation, we build a deep recommendation model called Spectral Collaborative Filtering (SpectralCF). Benefiting from the rich information of connectivity existing in the \textit{spectral domain}, SpectralCF is capable of discovering deep connections between users and items and therefore, alleviates the \textit{cold-start} problem for CF. To the best of our knowledge, SpectralCF is the first CF-based method directly learning from the \textit{spectral domains} of user-item bipartite graphs. We apply our method on several standard datasets. It is shown that SpectralCF significantly outperforms state-of-the-art models. Code and data are available at \url{https://github.com/lzheng21/SpectralCF}.%publicly released.\footnote{\url{https://github.com/anonymous121212/SpectralCF}}
\end{abstract}
\copyrightyear{2018} 
\acmYear{2018} 
\setcopyright{acmcopyright}
\acmConference[RecSys '18]{Twelfth ACM Conference on Recommender Systems}{October 2--7, 2018}{Vancouver, BC, Canada}
\acmBooktitle{Twelfth ACM Conference on Recommender Systems (RecSys '18), October 2--7, 2018, Vancouver, BC, Canada}
\acmPrice{15.00}
\acmDOI{10.1145/3240323.3240343}
\acmISBN{978-1-4503-5901-6/18/10}




\maketitle

\input{body}

\bibliographystyle{ACM-Reference-Format}
\bibliography{bib.bib}

\end{document}

%% file: body.tex
\section{Introduction}
\label{sec:intro}
%In the vast ocean of information on the Internet, massive amounts of items have been produced every day. Thus, customers need new tools to manage what is available to them. With the ability to provide personalized suggestions, RS, which are systems that make item recommendations for users based on users' interests, have become an important tool in many web services for attracting and retaining users.
The effectiveness of recommender systems (RS) often relies on how well users' interests or preferences can be understood and interactions between users and items can be modeled.
Collaborative Filtering (CF) \cite{koren2009matrix} is one of the widely used and prominent techniques for RS. The underlying assumption of the CF approach is that if a user $u_1$ shares a common item with another user $u_2$, $u_1$ is also likely to be interested in other items liked by $u_2$. Although CF has been successfully applied to many recommendation applications, the \textit{cold-start} problem is considered as one of its major challenges \cite{koren2009matrix}. The problem arises when a user interacted with a very small number of items. Consequently, the user shares few items with other users, and effectively recommending for the user becomes a challenging task for RS.
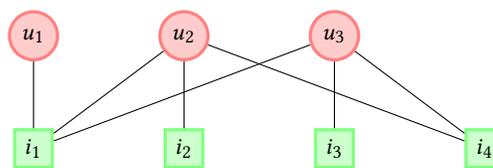
\begin{figure}[t]
\begin{tikzpicture}
 [L1Node/.style={circle,draw=red!50, fill=red!20, very thick, minimum size=5mm},
 L2Node/.style={rectangle,draw=green!50,fill=green!20,very thick, minimum size=5mm}]
       \node[L1Node] (u1) at (2, 2.5){$u_1$};
       %\node[L2Node] (i1) at (1, 1){$i_1$};
       \node[L1Node] (u2) at (4, 2.5){$u_2$};
       \node[L1Node] (u3) at (6, 2.5){$u_3$};
       %\node[L1Node] (u4) at (8, 2.5){$u_4$};
       
      \node[L2Node] (i1) at (2, 1){$i_1$};
       \node[L2Node] (i2) at (4, 1){$i_2$};
       \node[L2Node] (i3) at (6, 1){$i_3$};
       \node[L2Node] (i4) at (8, 1){$i_4$};
       %\draw [black] (u1)--(i1);
       \draw [black] (u1)--(i1);
       \draw [black] (u2)--(i1);
       \draw [black] (u2)--(i2);
       \draw [black] (u3)--(i1);
       \draw [black] (u3)--(i3);
       \draw [black] (u2)--(i4);
       \draw [black] (u3)--(i4);
       %\draw [black] (u4)--(i2);
       %\draw [black] (u4)--(i3);
       %\draw [black] (u4)--(i4);
  \end{tikzpicture}
  \caption{A toy example of a user-item bipartite graph $\mathcal{B}$ with edges representing observed user-item interactions. Red circles and green rectangles denote users and items, respectively.}
  \vspace{-1em}
   \label{fig:example}
\end{figure}

If we formulate the relationships between users and items as a bipartite graph\footnote{In this paper, we use the terminology "graph" to refer to the
graph/network structure of data and "network" for the architecture of machine learning models.}, we argue that the connectivity information of the graph can play an important role for tackling the \textit{cold-start} problem. For example, let us see a bipartite graph $\mathcal{B}$ in Figure \ref{fig:example}. A \textit{cold-start} user $u_1$ only interacts with item $i_1$. Since $u_1$ shares $i_1$ with user $u_2$ and user $u_3$, as a result, three items ($i_2$, $i_3$ and $i_4$) connected with $u_2$ or $u_3$ can all be recommended to $u_1$ by a CF-based model. However, a natural and important question arises: which one in the three items is the most reliable recommendation for $u_1$? The key to answer the question lies in the user-item connectivity information. In fact, if we take a look at the connections of the graph, it is clear that there is only one path existing between $u_1$ and $i_2$ (or $i_3$), while two paths connect $u_1$ to $i_4$. Thus, compared with $i_2$ and $i_3$, obviously, $i_4$ is a more reliable recommendation for $u_1$.   %But, compared with $i_2$ and $i_3$, $i_4$ is harder for traditional CF methods to discover. What is worse is that if the distance between $d(u_1$ and $i_4)$ becomes larger, despite a number of paths existing between $u_1$ and $i_4$, it is a formidable task for traditional CF models to uncover the deep connection between $u_1$ and $i_4$.

However, existing CF-based methods, including model-based and memory-based approaches, often suffer from the difficulty of modeling the connectivity information. Previous model-based approaches, such as Matrix Factorization (MF) \cite{koren2009matrix}, are usually designed to approximate the direct connections (or proximities). However, indirect connectivity information hidden in the graph structures is rarely captured by traditional model-based approaches. For instance, it is formidable for them to model the number of paths between $u_1$ and $i_4$ in Figure \ref{sec:intro}. Whereas a number of memory-based approaches \cite{sarwar2001item,jamali2009trustwalker} is introduced to model the connectivity information, these methods often rely on pre-defined similarity functions. However, in the real world, defining an appropriate similarity function suitable for diverse application cases is never an easy task.

\textit{Spectral graph theory} \cite{shuman2013emerging} studies connections between combinatorial properties of a graph and the eigenvalues of matrices associated to the graph, such as the laplacian matrix (see Definition \ref{def::lap} in Section \ref{sec:prelim}). In general, the spectrum of a graph focuses on the connectivity of the graph, instead of the geometrical proximity. %The theory about their eigenvalues and eigenvectors has been the object of studying for more than forty years. As it turns out, the spectral perspective is a powerful tool and exhibits many important properties of a graph. For example, \cite{brouwer2011spectra} shows that the multiplicity of zero eigenvalues equals the number of connected components of the graph and the second smallest eigenvalue $\lambda_2 > 0$ if and only if the graph is connected. In general, the spectrum of a graph focuses the connectivity of the graph instead of the geometrical proximity. \textit{Spectral graph theory} can be employed to estimate (or sometimes determine exactly) many properties of a graph, such as the best split of the graph or how fast information, diseases or trends spread on the graph \cite{spielman2007spectral}. 
To see how does the \textit{spectral domain} come to help for recommendations and better understand the advantages of viewing a user-item bipartite graph in the spectral perspective, let us revisit the toy example shown in Figure \ref{fig:example}. For the bipartite graph $\mathcal{B}$, we visually plot its vertices in one specific frequency domain. Although vertices do not come with coordinates, a popular way to draw them in a space is to use eigenvectors of a laplacian matrix associated with the graph to supply coordinates \cite{spielman2007spectral}. Figure \ref{fig:plot} shows that, compared with $i_2$ and $i_3$, $i_4$ becomes closer to $u_1$ in the space\footnote{In \textit{spectral graph theory}, smaller (or larger) eigenvalues of the associated laplacian matrix corresponds to lower- (or higher-) frequency domains. In Figure \ref{fig:example}, we plot each vertex $j$ at the point $(\boldsymbol{\mu}_1(j), \boldsymbol{\mu}_2(j))$, where $\boldsymbol{\mu}_l(j)$ indicates the $j_\textrm{th}$ value of the $l_\textrm{th}$ eigenvector of the laplacian matrix $\boldsymbol{L}$.}. Thus, when transformed into the frequency domain, $i_4$ is revealed to be a more suitable choice than $i_2$ or $i_3$ for $u_1$. The underlying reason is that the connectivity information of the graph has been uncovered in the frequency domain, where the relationships between vertices depend on not only their proximity but also connectivity. Thus, exploiting the spectrum of a graph can help better explore and identify the items to be recommended.

Inspired by the recent progress \cite{kipf2016semi,defferrard2016convolutional} in node/graph classification methods, we propose a \textit{spectral graph theory} based method to leverage the broad information existing in the \textit{spectral domain} to overcome the aforementioned drawbacks and challenges. Specifically, to conquer the difficulties (see Section \ref{sec::poly}) of directly learning from the \textit{spectral domain} for recommendations, we first present a new spectral convolution operation (see Eq.~(\ref{eq::poly_filter2})), which is approximated by a polynomial to dynamically amplify or attenuate each frequency domain. Then, we introduce a deep recommendation model, named Spectral Collaborative Filtering (SpectralCF), built by multiple proposed spectral convolution layers. SpectralCF directly performs collaborative filtering in the \textit{spectral domain}.
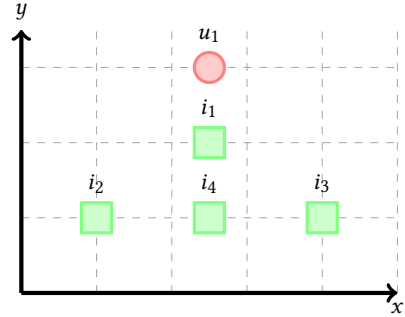
\begin{figure}
\begin{tikzpicture}
[L1Node/.style={circle,draw=red!50, fill=red!20, very thick, minimum size=4mm},
 L2Node/.style={rectangle,draw=green!50,fill=green!20,very thick, minimum size=4mm}]
 %\draw[help lines, color=gray!70, dashed] (0,0) grid (3.7,3.5);
  \draw[help lines, color=gray!70, dashed] (2,0) grid (7,3.5);

%\draw[->,ultra thick] (0,0)--(3.7,0) node[below]{$x$};
%\draw[->,ultra thick] (0,0)--(0,3.5) node[above]{$y$};
\draw[->,ultra thick] (2,0)--(7,0) node[below]{$x$};
\draw[->,ultra thick] (2,0)--(2,3.5) node[above]{$y$};

 %\node[L1Node][label={$u_1$}] (u1) at (0.5, 1.8){};
  %\node[L2Node][label={$i_1$}] (i1) at (1.2, 1.8){};
  %\node[L2Node][label={$i_2$}] (i2) at (2.1, 1.2){};
  %\node[L2Node][label={$i_3$}] (i3) at (2.1, 2.6){};
  %\node[L2Node][label={$i_4$}] (i4) at (3.0, 1.8){};
  
  \node[L1Node][label={$u_1$}] (u1) at (4.5, 3){};
  \node[L2Node][label={$i_2$}] (i2) at (3, 1){};
  \node[L2Node][label={$i_1$}] (i1) at (4.5, 2){};
    \node[L2Node][label={$i_3$}] (i2) at (6, 1){};
    \node[L2Node][label={$i_4$}] (i4) at (4.5, 1){};    
%\node[fill=white](why1) at (2,-0.6) {(a) Low-frequency Domain};
%\node[fill=white](why2) at (6,-0.6) {(b) High-frequency Domain};
\end{tikzpicture} 
\vspace{-1em}
\caption{Vertices of the bipartite graph in Figure \ref{fig:example} are plotted in a frequency domain. Note that the vertices not shown above are omitted for simplicity.}
\vspace{-1em}
\label{fig:plot}
\end{figure}

The key contributions of this work can be summarized as follows:
 \begin{itemize}
 \item \textbf{Novelty}: To the best of our knowledge, it is the first CF-based method directly learning from the \textit{spectral domains} of user-item bipartite graphs.
 \item \textbf{A deep recommendation model}: We propose a new spectral convolution operation performing in the \textit{spectral domain}. Stacked by multiple layers of the proposed spectral convolution operation, a deep recommendation model, named Spectral Collaborative Filtering (SpectralCF), is introduced. 
 \item \textbf{Strong Performance}: In the experiments, SpectralCF outperforms state-of-the-art comparative models. It is shown that SpectralCF effectively utilizes the rich information of connectivity existing in the \textit{spectral} domain to ease the \textit{cold-start} problem.
\end{itemize}
  
The rest of the paper is organized as follows. In Section \ref{sec:prelim}, we provide preliminary concepts. Section \ref{sec:model} describes SpectralCF in detail. Experiments are presented in Section \ref{sec:exp} to analyze SpectralCF and demonstrate its effectiveness compared with state-of-the-art techniques for RS. In Section \ref{sec:related}, we give a short review of the works related to our study. Finally, conclusions are presented in Section \ref{sec:con}.

\section{Definitions and Preliminaries}
\label{sec:prelim}
In this section, we present the background and preliminaries of this study. Throughout the paper, we denote scalars by either lowercase or uppercase letters, vectors by boldfaced lowercase letters, and matrices by boldfaced uppercase letters. Unless otherwise specified, all vectors are considered to be column vectors. Let $\boldsymbol{I}$ denote an identity matrix, and $\boldsymbol{1}$ and $\boldsymbol{0}$ denote matrices of ones and zeros, respectively. In addition, we define the following definitions in this paper as:
\begin{definition}
\textit{(Bipartite Graph). A bipartite user-item graph with $N$ vertices and $E$ edges for recommendations is defined as $\mathcal{B} = \{\mathcal{U}, \mathcal{I}, \mathcal{E}\}$, where $\mathcal{U}$ and $\mathcal{I}$ are two disjoint vertex sets of users and items. Every edge $e \in \mathcal{E}$ has the form $e=(u,i)$ where $u \in \mathcal{U}$ and $i \in \mathcal{I}$ and denotes that user $u$ has interacted with item $i$ in the training set}. 
\end{definition}
\begin{definition}
\textit{(Implicit Feedback Matrix). An implicit feedback matrix $\boldsymbol{R}$ is a $|\mathcal{U}| \times |\mathcal{I}|$ matrix defined as following:
\begin{equation}
\boldsymbol{R}_{r,j} = \left\{\begin{array}{cc}
    1 &\text{if } (u_r,i_j) \text{ interaction is observed}, \\
    0 &\text{otherwise}.\\
  \end{array} \right.
\end{equation}  }
\end{definition}
\begin{definition}
\textit{(Adjacent Matrix). For the bipartite graph $\mathcal{B}$, its corresponding adjacent matrix $\boldsymbol{A}$ can be defined as:
\begin{equation}
\boldsymbol{A} = \left[
\begin{array}{cc}
\boldsymbol{0} & \boldsymbol{R} \\
\boldsymbol{R}^\intercal & \boldsymbol{0}
\end{array}
\right],
\end{equation}
where $\boldsymbol{A}$ is an $N \times N$ matrix. }
\end{definition}
\begin{definition}
\label{def::lap}
\textit{(Laplacian Matrix). The random walk laplacian matrix $\boldsymbol{L}$ is defined as $\boldsymbol{L} = \boldsymbol{I} - \boldsymbol{D^{-1}\boldsymbol{A}}$, where $\boldsymbol{I}$ is the $N \times N$ identity matrix and $\boldsymbol{D}$ is the $N \times N$ diagonal degree matrix defined as $D_{nn}=\sum_{j}{A_{n,j}}$.}
%\begin{equation}
%\boldsymbol{D}_{ij} = \left\{\begin{array}{cc}
%    \sum_p{R_{ip}} &i=j, \\
%    0 &otherwise.\\
%  \end{array} \right.
%\end{equation}
\end{definition}
This paper focuses on the recommendation problem with implicit feedbacks, where we only observe whether a person has viewed/liked/clicked an item and do not observe explicit ratings. Let $\mathcal{I}_i^+$ denote the set of all items liked by user $i$ and $\mathcal{I}^-_i$ denote the remaining items. We define the recommendation problem which we study in this paper as the following:
\begin{definition}
\textit{(Problem Definition). Given a user set $\mathcal{U}$ and an item set $\mathcal{I}$, for each user $u \in \mathcal{U}$ who has liked/clicked/viewed an item set $\mathcal{I}^+_u \subseteq \mathcal{I}$, we aim to recommend a ranked list of items from $\mathcal{I}_u^-$ that are of interests to the user.}
\end{definition}

\section{Proposed Model}
In this section, we first describe the process of performing a \textit{graph fourier transform} on a bipartite graph $\mathcal{B}$ for recommendations. Then we propose to place a novel spectral convolution filter on vertices (users and items) of the bipartite graph to dynamically filter the contributions of each frequency component in the \textit{spectral domain}. Later, a polynomial approximation is employed to overcome the shortcomings of the proposed convolution operation. Finally, with the approximate convolution operation, we introduce our final recommender system, named Spectral Collaborative Filtering, stacked by multiple spectral convolution layers.
\label{sec:model}
\subsection{Graph Fourier Transform}
\begin{definition}
\label{def::signal}
\textit{(Graph Signal). Given any graph $\mathcal{G}=\{\mathcal{V},\mathcal{E}\}$, where $\mathcal{V}$ and $\mathcal{E}$ are a vertex and an edge set, respectively, a \textit{graph signal} is defined as a state vector $\boldsymbol{x} \in \mathcal{R}^{|\mathcal{V}|\times 1}$ over all vertices in the graph, where $x_j$ is the $j_\textrm{th}$ value of $\boldsymbol{x}$ observed at the $j_\textrm{th}$ vertex of $\mathcal{G}$.}
\end{definition}
%Suppose we observe two signals, $\boldsymbol{x}^u \in\mathcal{R}^{|\mathcal{U}|\times 1}$ and $\boldsymbol{x}^i\in\mathcal{R}^{|\mathcal{I}|\times 1}$, associated respectively with $\mathcal{U}$ and $\mathcal{V}$ of the bipartite graph $\mathcal{B}$ for recommendations, we are interested in processing the two signals in this subsection. 

The classical \textit{fourier transform} is defined as an expansion of a function $f$ in terms of the complex exponentials as:
\begin{equation}
\label{eq::fourier}
\hat{f}(\xi)=\int_{-\infty}^{+\infty} f(x)e^{-2\pi \mathit{i}\xi}dx,
\end{equation}
where $\mathit{i}$ is an imaginary number, and the complex exponentials ($e^{-2\pi \mathit{i}\xi}$) form an orthonormal basis. 

Analogously, the \textit{graph fourier transform} is defined as an expansion of an observed \textit{graph signal} in terms of the eigenvectors of the graph laplacian $\boldsymbol{L}$, and the eigenvectors serve as a basis in the \textit{spectral domain}. Let us assume that a \textit{graph signal} ($\boldsymbol{x} \in \mathcal{R}^{|\mathcal{V}|\times 1}$) is observed on a graph $\mathcal{G}$, we define the \textit{graph fourier transform} and its inverse on $\mathcal{G}$ as:
\begin{equation}
\label{transform}
\hat{x}(l)=\sum_{j=0}^{N-1}x(j)\mu_l(j) \quad\text{ and }\quad x(j)=\sum_{l=0}^{N-1}\hat{x}(l)\mu_l(j),
\end{equation}
where $x(j)$, $\hat{x}(l)$ and $\mu_l(j)$ denote the $j_\textrm{th}$, $l_\textrm{th}$ and $j_\textrm{th}$ value of $\boldsymbol{x}$, $\hat{\boldsymbol{x}}$ and $\boldsymbol{\mu}_l$, respectively; $\boldsymbol{\mu}_l$ denotes the $l_\textrm{th}$ eigenvector of $\boldsymbol{L}$; $\hat{\boldsymbol{x}}$ represents a \textit{graph signal} which has been transformed into the \textit{spectral domain}. For simplicity, we rewrite Eq.~(\ref{transform}) in the matrix form as $\hat{\boldsymbol{x}}=\boldsymbol{U}^\intercal\boldsymbol{x}$ and $\boldsymbol{x}=\boldsymbol{U}\hat{\boldsymbol{x}}$, respectively, where $\boldsymbol{U} = \{\boldsymbol{\mu}_0,\boldsymbol{\mu}_1,...,\boldsymbol{\mu}_l,...,\boldsymbol{\mu}_{N-1}\}$ are eigenvectors of $\boldsymbol{L}$. 

In particular, for a bipartite graph $\mathcal{B}$, assume that there are two types of \textit{graph signals}: $\boldsymbol{x}^u\in\mathcal{R}^{|\mathcal{U}|\times 1}$ and $\boldsymbol{x}^i\in\mathcal{R}^{|\mathcal{I}|\times 1}$, associated with user and item vertices, respectively. We transform them into the \textit{spectral domain} and vice versa as :
\begin{equation}
\left[
\begin{array}{c}
\hat{\boldsymbol{x}}^u  \\
\hat{\boldsymbol{x}}^i 
\end{array}
\right]=\boldsymbol{U}^\intercal\left[
\begin{array}{c}
\boldsymbol{x}^u  \\
\boldsymbol{x}^i 
\end{array}
\right] \quad\text{and}\quad \left[
\begin{array}{c}
\boldsymbol{x}^u  \\
\boldsymbol{x}^i 
\end{array}
\right]=\boldsymbol{U}\left[
\begin{array}{c}
\hat{\boldsymbol{x}}^u  \\
\hat{\boldsymbol{x}}^i 
\end{array}
\right].
\end{equation}
\subsection{Spectral Convolution Filtering}
The broad information of graph structures exists in the \textit{spectral domain}, and different types of connectivity information between users and items can be uncovered in different frequency domains. It is desirable to dynamically adjust the importance of each frequency domain for RS.

To this end, we propose a convolution filter, parameterized by $\boldsymbol{\theta}\in \mathcal{R}^N$, as $g_{\boldsymbol{\theta}}(\boldsymbol{\Lambda})=diag([\theta_0\lambda_0,\theta_1\lambda_1,...,\theta_{N-1}\lambda_{N-1}])$ into the  \textit{spectral domain} as: 
\begin{equation}
\label{eq::filter}
\left[
\begin{array}{c}
\boldsymbol{x}^u_{new}  \\
\boldsymbol{x}^i_{new} 
\end{array}
\right]=\boldsymbol{U}g_{\boldsymbol{\theta}}(\boldsymbol{\Lambda})\left[
\begin{array}{c}
\hat{\boldsymbol{x}}^u  \\
\hat{\boldsymbol{x}}^i 
\end{array}
\right]=\boldsymbol{U}g_{\boldsymbol{\theta}}(\boldsymbol{\Lambda})\boldsymbol{U}^\intercal\left[
\begin{array}{c}
\boldsymbol{x}^u  \\
\boldsymbol{x}^i 
\end{array}
\right],
\end{equation}
where $\boldsymbol{x}^u_{new}$ and $\boldsymbol{x}^i_{new}$ are new \textit{graph signals} on $\mathcal{B}$ learned by the filter $g_{\boldsymbol{\theta}}(\boldsymbol{\Lambda})$, and $\boldsymbol{\Lambda}=\{\lambda_0,\lambda_1,...,\lambda_{N-1}\}$ denotes eigenvalues of the graph laplacian matrix $\mathbf{L}$. 

In Eq.~(\ref{eq::filter}), a convolution filter $g_{\boldsymbol{\theta}}(\boldsymbol{\Lambda})$ is placed on a spectral \textit{graph signal} $\left[
\begin{array}{c}
\hat{\boldsymbol{x}}^u  \\
\hat{\boldsymbol{x}}^i 
\end{array}
\right]$, and each value of $\boldsymbol{\theta}$ is responsible for boosting or diminishing each corresponding frequency component. The eigenvector matrix $\boldsymbol{U}$ in Eq.~(\ref{eq::filter}) is used to perform an inverse \textit{graph fourier transform}.
\subsection{Polynomial Approximation}
\label{sec::poly}
Recall that we proposed a convolution operation, as shown in Eq.~(\ref{eq::filter}), to directly perform in the \textit{spectral domain}. Although the filter is able to dynamically measure contributions of each frequency component for the purpose of recommendations, there are two limitations. First, as shown in Eq.~(\ref{eq::filter}), the learning complexity of the filter is $\mathcal{O}(N)$, where $N$ is the number of vertices. That is, unlike classical Convolutional Neural Networks (CNNs), the number of parameters of the filter is linear to the dimensionality of data. It constrains the scalability of the proposed filter. Second, the learned \textit{graph signals} ($\boldsymbol{x}^u_{new} \in \mathcal{R}^{|\mathcal{U}|\times 1}$ and $\boldsymbol{x}^i_{new} \in \mathcal{R}^{|\mathcal{I}|\times 1}$) are vectors. It means that each vertex of users or items is represented by a scalar feature. However, a vector for every user and item is necessary to model the deep and complex connections between users and items.
\begin{figure*}[t]
\label{fig:feedforward}
\includegraphics[width=0.9\textwidth,height=0.12\textheight]{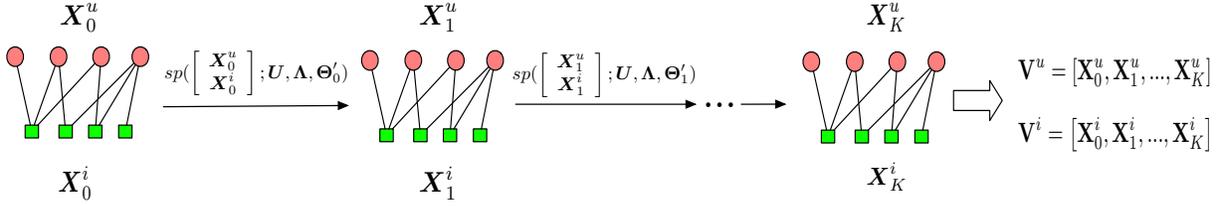}
\vspace{-1em}
\caption{The feed-forward procedure of SpectralCF. The function $sp(:;\boldsymbol{U},\boldsymbol{\Lambda},\boldsymbol{\Theta})$ denotes the spectral convolution operation shown in Eq.~(\ref{eq::poly_filter2}).}
\vspace{-1em}
\end{figure*}

The first limitation can be overcome by using a polynomial approximation. We first demonstrate that the set of all convolution filters $\mathcal{S}_g=\{g_{\boldsymbol{\theta}}(\boldsymbol{\Lambda})=diag([\theta_0\lambda_0,\theta_1\lambda_1,...,\theta_{N-1}\lambda_{N-1}]),\boldsymbol{\theta}\in \mathcal{R}^N\}$ is equal to the set of finite-order polynomials $\mathcal{S}_h=\{h_{\boldsymbol{\theta'}}(\boldsymbol{\Lambda})={\sum\limits_{p=0}^{N-1}\theta'_p\boldsymbol{\Lambda}^p},\boldsymbol{\theta'}\in \mathcal{R}^N\}$.
\begin{prop}
\label{pro::polynomials}
$\mathcal{S}_h$ is equal to $\mathcal{S}_g$.
\end{prop}
\begin{proof}
Let us consider an instance $h_{\boldsymbol{\theta'}}(\boldsymbol{\Lambda}) \in \mathcal{S}_h$. Then, $h_{\boldsymbol{\theta'}}(\boldsymbol{\Lambda})={\sum\limits_{p=0}^{N-1}\theta'_p\boldsymbol{\Lambda}^p}=diag([\sum\limits_{p=0}^{N-1}\theta'_p\lambda_0^{p-1}\cdot \lambda_{0},\sum\limits_{p=0}^{N-1}\theta'_p\lambda_1^{p-1}\cdot \lambda_{1},...,\sum\limits_{p=0}^{N-1}\theta'_p\lambda_{N-1}^{p-1}\cdot \lambda_{N-1}])$. So, $h_{\boldsymbol{\theta'}}(\boldsymbol{\Lambda}) \in  \mathcal{S}_g$. Now, consider a convolution filter $g_{\boldsymbol{\theta}}(\boldsymbol{\Lambda})\in \mathcal{S}_g$. Then, there must exist a polynomial function $\phi(\lambda)=\sum\limits_{p=0}^{N-1}a_p\lambda^p$ that interpolates through all pairs $(\lambda_i, \theta_i \lambda_i)$ for $i\in \{0,1,...,N-1\}$. The maximum degree of such a polynomial is at most $N-1$ as there are maximum $N$ points to interpolate. Therefore, $g_{\boldsymbol{\theta}}(\boldsymbol{\Lambda})={\sum\limits_{p=0}^{N-1}a_p\boldsymbol{\Lambda}^p}=h_{\boldsymbol{a}}(\boldsymbol{\Lambda})\in \mathcal{S}_h$.
\end{proof}

Now, we can approximate the convolution filters by using first $P$ polynomials as the following:
\begin{eqnarray}
\label{poly}
g_{\boldsymbol{\theta}}(\boldsymbol{\Lambda}) \approx \sum_{p=0}^P\theta'_p\boldsymbol{\Lambda}^p.
\end{eqnarray}

%The first limitation can be overcome by using a polynomial approximation \cite{defferrard2016convolutional} as the following:
%\begin{eqnarray}
%\label{poly}
%g_{\boldsymbol{\theta}}(\boldsymbol{\Lambda}) \approx \sum_{p=0}^P\theta_p\boldsymbol{\Lambda}^p.
%\end{eqnarray}
In this way, the learning complexity of the filter becomes $\mathcal{O}(P)$, where $P$ is a hyper-parameter, and independent from the number vertices. Specially, we limit the order of the polynomial, $P$, to 1 in order to avoid over-fitting. By substituting Eq.~(\ref{poly}) into Eq.~(\ref{eq::filter}), we have:
\begin{equation}
\label{eq::poly_filter}
\left[
\begin{array}{c}
\boldsymbol{x}^u_{new}  \\
\boldsymbol{x}^i_{new} 
\end{array}
\right]=(\theta'_0\boldsymbol{U}\boldsymbol{U}^\intercal+\theta'_1\boldsymbol{U}\boldsymbol{\Lambda}\boldsymbol{U}^\intercal)\left[
\begin{array}{c}
\boldsymbol{x}^u  \\
\boldsymbol{x}^i 
\end{array}
\right].
\end{equation}
Furthermore, it is beneficial to further decrease the number of parameters by setting $\theta'=\theta'_0=\theta'_1$. As a result, Eq.~(\ref{eq::poly_filter}) becomes:
\begin{equation}
\label{eq::poly_filter1}
\left[
\begin{array}{c}
\boldsymbol{x}^u_{new}  \\
\boldsymbol{x}^i_{new} 
\end{array}
\right]=\theta'(\boldsymbol{U}\boldsymbol{U}^\intercal+\boldsymbol{U}\boldsymbol{\Lambda}\boldsymbol{U}^\intercal)\left[
\begin{array}{c}
\boldsymbol{x}^u  \\
\boldsymbol{x}^i 
\end{array}
\right],
\end{equation}
where $\theta'$ is a scalar.

For the second limitation, one can generalize the \textit{graph signals} ($\boldsymbol{x}^u \in \mathcal{R}^{|\mathcal{U}| \times 1}$ and $\boldsymbol{x}^i \in \mathcal{R}^{|\mathcal{I}|\times 1}$) to $C$-dimensional \textit{graph signals}: $\boldsymbol{X}^u \in \mathcal{R}^{|\mathcal{U}|\times C}$ and $\boldsymbol{X}^i \in \mathcal{R}^{|\mathcal{I}|\times C}$. Hence, Eq.~(\ref{eq::poly_filter1}) becomes $\left[
\begin{array}{c}
\boldsymbol{X}^u_{new}  \\
\boldsymbol{X}^i_{new} 
\end{array}
\right]=(\boldsymbol{U}\boldsymbol{U}^\intercal+\boldsymbol{U}\boldsymbol{\Lambda}\boldsymbol{U}^\intercal)\left[
\begin{array}{c}
\boldsymbol{X}^u  \\
\boldsymbol{X}^i 
\end{array}
\right]\theta'$. To take one step further, we generalize the filter parameter $\theta'$ to a matrix of filter parameters $\boldsymbol{\Theta}' \in \mathcal{R}^{C \times F}$ with $C$ input channels and $F$ filters. As a result, our final spectral convolution operation is shown as the following:
\begin{equation}
\label{eq::poly_filter2}
\left[
\begin{array}{c}
\boldsymbol{X}^u_{new}  \\
\boldsymbol{X}^i_{new} 
\end{array}
\right]=\sigma \left(\left(\boldsymbol{U}\boldsymbol{U}^\intercal+\boldsymbol{U}\boldsymbol{\Lambda}\boldsymbol{U}^\intercal\right)\left[
\begin{array}{c}
\boldsymbol{X}^u  \\
\boldsymbol{X}^i 
\end{array}
\right]\boldsymbol{\Theta}'\right),
\end{equation}
where $\boldsymbol{X}^u_{new} \in \mathcal{R}^{|\mathcal{U}| \times F}$ and $\boldsymbol{X}^i_{new}\in \mathcal{R}^{|\mathcal{I}| \times F}$ denote convolution results learned with $F$ filters from the \textit{spectral domain} for users and items, respectively; $\sigma$ denotes the logistic sigmoid function.

In fact, Eq.~(\ref{eq::poly_filter2}) is a general version of Eq.~(\ref{eq::poly_filter1}) as it is equivalent to perform Eq.~(\ref{eq::poly_filter1}) in $C$ input channels with $F$ filters. Hereafter, the proposed convolution operation as shown in Eq.~(\ref{eq::poly_filter2}) is denoted as a function $sp(:;\boldsymbol{U},\boldsymbol{\Lambda},\boldsymbol{\Theta}')$, which is parameterized by $\boldsymbol{U},\boldsymbol{\Lambda}$ and $\boldsymbol{\Theta}'$.
\subsection{Multi-layer Model}
Given user vectors $\boldsymbol{X}^u$ and item vectors $\boldsymbol{X}^i$, new \textit{graph singals} ($\boldsymbol{X}^u_{new}$ and $\boldsymbol{X}^i_{new}$) in Eq.~(\ref{eq::poly_filter2}) are convolution results learned from the \textit{spectral domain} with a parameter matrix $\boldsymbol{\Theta}' \in \mathcal{R}^{C \times F}$. As in classical CNNs, one can regard Eq.~(\ref{eq::poly_filter2}) as a propagation rule to build a deep neural feed-forward network based model, which we refer as Spectral Collaborative Filtering (SpectralCF).

Similar to word embedding techniques, we first randomly initialize user vectors $\boldsymbol{X}^u_0$ and item vectors $\boldsymbol{X}^i_0$. Taking $\boldsymbol{X}^u_0$ and $\boldsymbol{X}^i_0$ as inputs, a $K$ layered deep spectralCF can be formulated as:
\begin{equation}
\label{eq::poly_filter3}
\left[
\begin{array}{c}
\boldsymbol{X}^u_{K}  \\
\boldsymbol{X}^i_{K} 
\end{array}
\right]=\underbrace{sp\Big(...sp\big(}_{K}\left[
\begin{array}{c}
\boldsymbol{X}^u_0  \\
\boldsymbol{X}^i_0 
\end{array}
\right];\boldsymbol{U},\boldsymbol{\Lambda},\boldsymbol{\Theta}'_0\big)...;\boldsymbol{U},\boldsymbol{\Lambda},\boldsymbol{\Theta}'_{K-1}\Big),
\end{equation}
where $\boldsymbol{\Theta}'_{K-1}\in\mathcal{R}^{F\times F}$ is a matrix of filter parameters for the $k_{\textrm{th}}$ layer; $\boldsymbol{X}^u_{k}$ and $\boldsymbol{X}^i_{k}$ denote the convolution filtering results of the $k_\textrm{th}$ layer.

In order to utilize features from all layers of SpectralCF, we further concatenate them into our final latent factors of users and items as:
\begin{equation}
\label{eq:concat}
\mathbf{V}^u = \left[\mathbf{X}^u_0,\mathbf{X}^u_1,...,\mathbf{X}^u_K\right]\quad\text{ and }\quad
\mathbf{V}^i = \left[\mathbf{X}^i_0,\mathbf{X}^i_1,...,\mathbf{X}^i_K\right],
\end{equation}
where $\mathbf{V}^u \in \mathcal{R}^{|\mathcal{U}|\times\left(C+KF\right)}$ and $\mathbf{V}^i \in \mathcal{R}^{|\mathcal{I}|\times\left(C+KF\right)}$.

In terms of the loss function, the conventional BPR loss suggested in \cite{rendle2009bpr} is employed. BPR is a pair-wise loss to address the implicit data for recommendations. Unlike point-wise based methods \cite{koren2008factorization}, BPR learns a triple $(r,j,j')$, where item $j$ is liked/clicked/viewed by user $r$ and item $j'$ is not. By maximizing the preference difference between $j$ and $j'$, BPR assumes that the user $i$ prefers item $j$ over the unobserved item $j'$. In particular, given a user matrix $\boldsymbol{V}^u$ and an item matrix $\boldsymbol{V}^i$ as shown in Eq.~(\ref{eq:concat}), the loss function of SpectralCF is given as: 
\begin{eqnarray}
\mathcal{L}=\text{arg }\underset{\boldsymbol{V}^u,\boldsymbol{V}^i}{\text{min}}  \sum_{(r,j,j') \in \mathcal{D}} -\text{ln } \sigma({\boldsymbol{v}^u_r}^\intercal\boldsymbol{v}^i_j - {\boldsymbol{v}^u_r}^\intercal\boldsymbol{v}^i_{j'}) \\ \nonumber
+ \lambda_{reg}(||\boldsymbol{V}^u||^2_2 + ||\boldsymbol{V}^i||^2_2),
\end{eqnarray}
where $\boldsymbol{v}^u_r$ and $\boldsymbol{v}^i_j$ denote $r_\textrm{th}$ and $j_\textrm{th}$ column of $\boldsymbol{V}^u$ and $\boldsymbol{V}^i$, respectively; $\lambda_{reg}$ represents the weight on the regularization terms. The training data $\mathcal{D}$ is generated as:
\begin{equation}
\mathcal{D} = \{(r,j,j')|r \in \mathcal{U} \wedge j \in \mathcal{I}_i^+ \wedge j' \in \mathcal{I}_i^-\}.
\end{equation}
\subsection{Optimization and Prediction} 
At last, RMSprop \cite{tieleman2012lecture} is used to minimize the loss function. The RMSprop is an adaptive version of gradient descent which adaptively controls the step size with respect to the absolute value of the gradient. It is done by scaling the updated value of each weight by a running average of its gradient norm. %The updating rules for all parameters of the proposed model, denoted as $\boldsymbol{\Psi}=\{\boldsymbol{\Theta}'_0,\boldsymbol{\Theta}'_1,...,\boldsymbol{\Theta}'_{K-1},\boldsymbol{X}^u_0,\boldsymbol{X}^i_0\}$, at epoch $e$ are estimated as:
 %\begin{equation}
 %{
 %\gamma^e \leftarrow 0.9{\left(\frac{\partial {\mathcal{L}}}{ \partial \boldsymbol{\Psi} ^e }\right)}^2 + 0.1\gamma^{e-1}, }
 %\label{eq:rmsprop1}
 %\end{equation} 
 %\begin{equation}
%{\boldsymbol{\Psi} ^{e+1} \leftarrow \boldsymbol{\Psi}^e - \left(\frac{\lambda}{\sqrt[]{\gamma^e}+\epsilon}\right) \frac{\partial \mathcal{L}}{ \partial \boldsymbol{\Psi} ^e },}
 %\label{eq:rmsprop2}
 %\end{equation}
%where $\lambda$ is the learning rate, and $\epsilon>0$ is a small positive value added for numerical stability. 

As shown in Algorithm \ref{alg}, for a batch of randomly sampled triple $(r,j,j')$, we update parameters in each epoch using the gradients of the loss function. After the training process, with optimized $\boldsymbol{\Theta}$, $\boldsymbol{X}^u_0$ and $\boldsymbol{X}^i_0$, we derive the user $r$'s preference over item $j$ as ${\boldsymbol{v}^u_r}^\intercal\boldsymbol{v}^i_j$. The final item recommendation for a user $r$ is given according to the ranking criterion as Eq.~(\ref{ranking}).
\begin{equation}
\label{ranking}
r: j_1 \succcurlyeq j_2 \succcurlyeq ... \succcurlyeq j_n \Rightarrow {\boldsymbol{v}^u_r}^\intercal\boldsymbol{v}^i_{j_1} > {\boldsymbol{v}^u_r}^\intercal\boldsymbol{v}^i_{j_2} > ... > {\boldsymbol{v}^u_r}^\intercal\boldsymbol{v}^i_{j_n}.
\end{equation}

\section{Experiments}
\label{sec:exp}
\begin{algorithm}[t]
 \caption{SpectralCF}
 \label{alg}
\SetAlgoLined {\small
\KwIn{Training set: $\mathcal{D} := \{(r,j,j')|r \in \mathcal{U} \wedge j \in \mathcal{I}^+_i \wedge j' \subseteq \mathcal{I}^-_i\}$,
number of epochs $E$,
batch size $B$,
number of layers $K$, dimension of latent factors $C$, number of filters $F$, regularization term $\lambda_{reg}$, learning rate $\lambda$, laplacian matrix $\boldsymbol{L}$ and its corresponding eigenvectors $\boldsymbol{U}$ and eigenvalues $\boldsymbol{\Lambda}$.}
\KwOut{Model's parameter set: $\boldsymbol{\Psi} =\{\boldsymbol{\Theta}'_0,\boldsymbol{\Theta}'_1,...,\boldsymbol{\Theta}'_{K-1},\boldsymbol{X}^u_0,\boldsymbol{X}^i_0$\}.}
Randomly initialize $\mathbf{X}^u_0$ and $\mathbf{X}^i_0$ from a Gaussian distribution $\mathcal{N}(0.01,0.02)$;\\
\For{${e}=1,2,\cdots,E$}{
	Generate the $e_\textrm{th}$ batch of size $B$ by uniformly sampling from $\mathcal{U}$, $\mathcal{I}_i^+$ and $\mathcal{I}_i^-$;\\
	\For{${k}=0,1,\cdots,K-1$}
    {
    Calculate $\boldsymbol{X}^u_{k+1}$ and $\boldsymbol{X}^i_{k+1}$ by using Eq.~(\ref{eq::poly_filter2});
    }
 % Randomly select a sample $p=(x_1,x_2)$ from training data $X$\;
Concatenate $[\boldsymbol{X}^u_{0},\boldsymbol{X}^u_{1},...,\boldsymbol{X}^u_{K}]$ into $\boldsymbol{V}^u$ and $[\boldsymbol{X}^i_{0},\boldsymbol{X}^i_{1},...,\boldsymbol{X}^i_{K}]$ into $\boldsymbol{V}^i$;\\
Estimate gradients $\frac{{\partial \mathcal{L}^{}}}{{\partial \mathbf{\Psi}_{e} }}$ by back propagation;\\
Update $\boldsymbol{\Psi}_{e+1}$ according to the procedure of RMSprop optimization \cite{tieleman2012lecture};
}
\KwRet{$\mathbf{\Psi}_E$}.}
\end{algorithm}
As discussed in the introduction section, leveraging the connectivity information in a user-item bipartite graph is essentially important for an effective recommendation model. In this section, we argue that, directly learning from the \textit{spectral domain}, the proposed SpectraCF can reveal the rich information of graph structures existing in the \textit{spectral domain} for making better recommendations. One may ask the following research questions:  
\begin{itemize}[leftmargin=0cm]
\item[] \textbf{RQ1}: How much does SpectralCF benefit from the connectivity information learned from the \textit{spectral domain}? 
\item[] \textbf{RQ2}: Does SpectralCF learn from the \textit{spectral domain} in an effective way?
\item[] \textbf{RQ3}: Compared with traditional methods, can SpectralCF better counter the \textit{cold-start} problem?  
\end{itemize}
In this section, in order to answer the questions above, we conduct experiments to compare SpectralCF with state-of-the-art models.
\subsection{Comparative Methods}
\label{sec:baseline}
To validate the effectiveness of SpectralCF, we compare it with six state-of-the-art models. The comparative models can be categorized into two groups: (1) \textbf{CF-based Models:} To answer \textbf{RQ1}, we compare SpectralCF with four state-of-the-art CF-based methods (ItemKNN, BPR, eALS and NCF) which ignore the information in the \textit{spectral domain}; (2) \textbf{Graph-based Models}: For \textbf{RQ2}, we are interested in how effectively does SpetralCF learn the connectivity information from the \textit{spectral domain}. We therefore compare SpectralCF with two graph-based models: GNMF and GCMC. Although the two models are also CF-based, we term them as graph-based models since they learn the structural information from a bipartite graph. These two groups of comparative models are summarized below:
\begin{itemize}
\item \textbf{ItemKNN} \cite{sarwar2001item}: ItemKNN is a standard neighbor-based collaborative filtering method. The model finds similar items for a user based on their similarities.
\item \textbf{BPR} \cite{rendle2009bpr}:  We use \textbf{B}ayesian \textbf{P}ersonalized \textbf{R}anking based Matrix Factorization. BPR introduces a pair-wise loss into the Matrix Factorization to be optimized for ranking \cite{Gantner2011MyMediaLite}.
\item \textbf{eALS} \cite{he2016fast}: This is a state-of-the-art matrix factorization based method for
item recommendation. This model takes all unobserved interactions as negative instances and weighting them non-uniformly by the item popularity.
\item \textbf{NCF} \cite{neuralCF}: \textbf{N}eural \textbf{C}ollaborative \textbf{F}iltering fuses matrix factorization and Multi-Layer Perceptron (MLP) to learn from user-item interactions. The MLP endows
NCF with the ability of modelling non-linearities between users and items.
\item \textbf{GNMF} \cite{cai2008non}: \textbf{G}raph regularized \textbf{N}on-negative \textbf{M}atrix \textbf{F}acto-rization considers the graph structures by seeking a matrix factorization with a graph-based regularization.

\item \textbf{GCMC} \cite{berg2017graph}: \textbf{G}raph \textbf{C}onvolutional \textbf{M}atrix \textbf{C}ompletion utilizes a graph auto-encoder to learn the connectivity information of a bipartite interaction graph for latent factors of users and items. 
\end{itemize}
Please note that, GNMF and GCMC are originally designed for explicit datasets. For a fair comparison, we follow the setting of \cite{hu2008collaborative} to adapt them for implicit data.
\subsection{Datasets}
\begin{figure}[t]
\centering
\includegraphics[width=0.23\textwidth,height=.15\textheight]{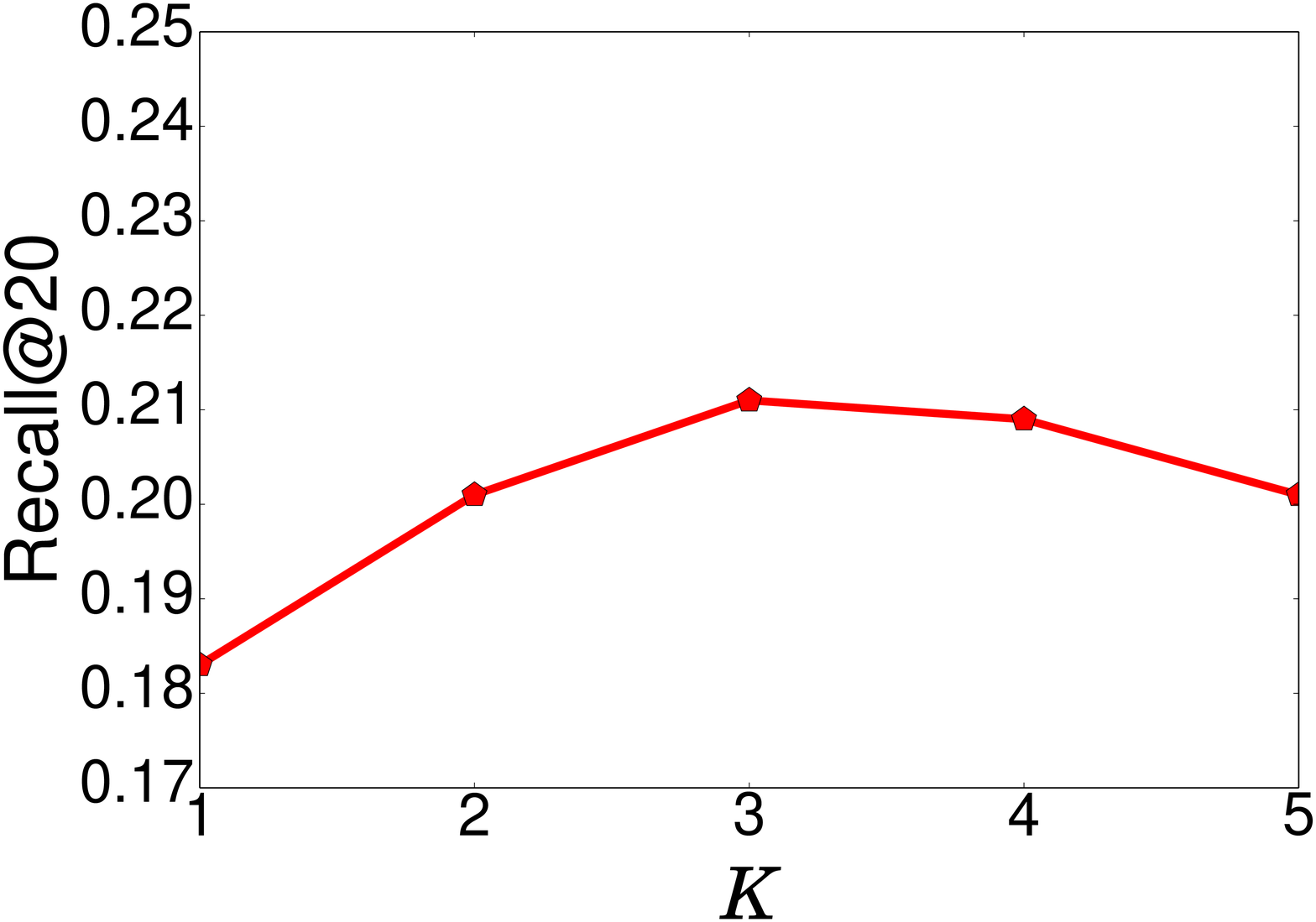}
\includegraphics[width=0.23\textwidth,height=.15\textheight]{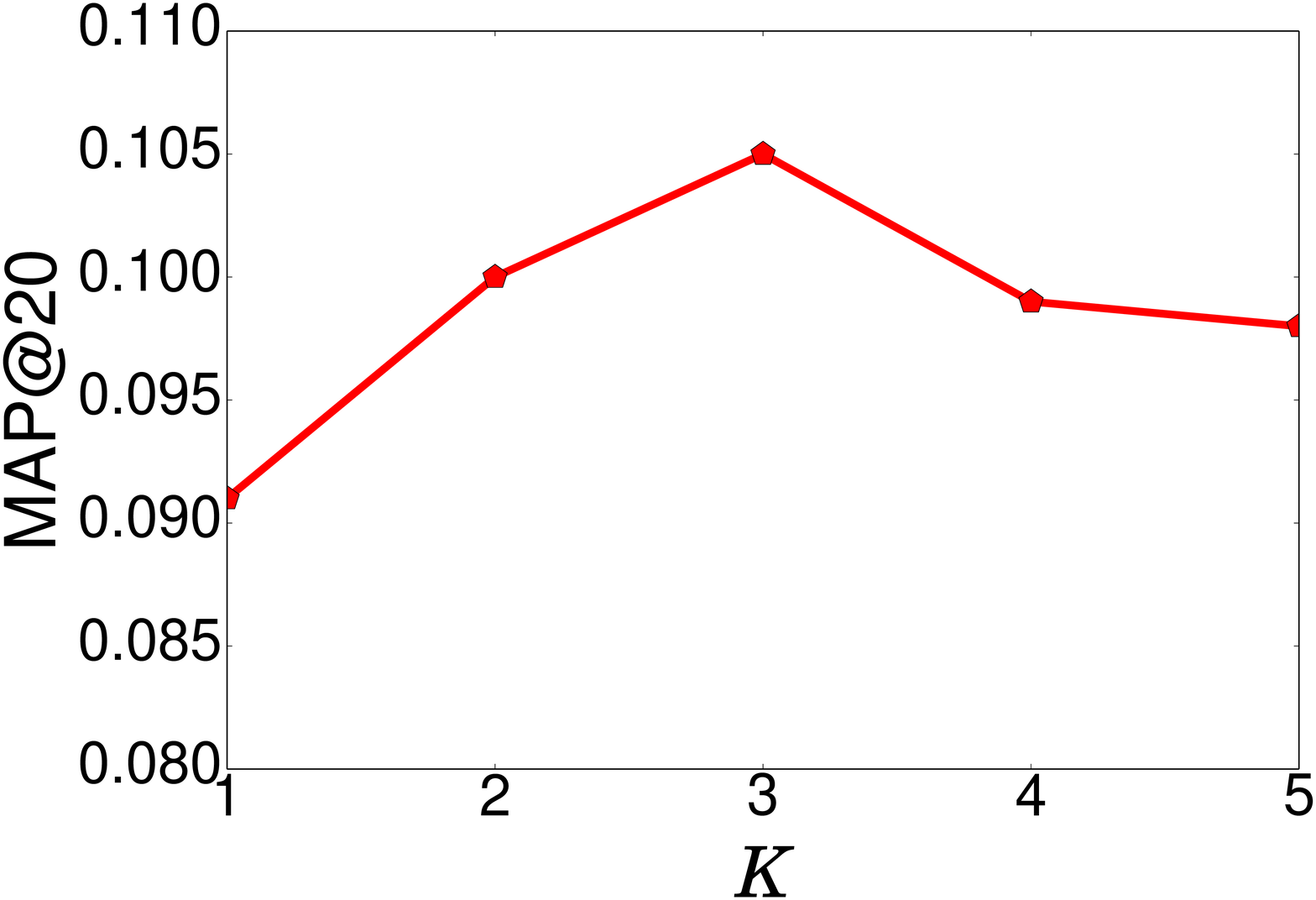}
\vspace{-1em}
\caption{Effects of hyper-parameter $K$ in terms of Recall@20 and MAP@20 in the dataset of \textit{MovieLens-1M}.}
\vspace{-1em}
\label{fig:：parameter}
\end{figure} 
\begin{table}[t]
\centering
\caption{The hyper-parameter setting of SpectralCF.}
\vspace{-1em}
\label{table::setting}
\begin{tabular}{|c||M{0.15cm}|M{0.15cm}|M{0.15cm}|c|c|c|c|}
\hline
Hyper-parameters &$K$&$C$&$F$ &$\lambda_{reg}$&$B$&$E$&$\lambda$\\
\hline
Values & 3&$16$&$16$ &0.001&$1,024$&$200$&0.001\\
\hline
\end{tabular}
\vspace{-1em}
\end{table}
We test our method as well as comparative models on three publicly available datasets\footnote{\textit{MovieLens-1M} and \textit{HetRec} are available at \url{https://grouplens.org/datasets/}; and \textit{Amazon Instant Video} can be found at \url{http://jmcauley.ucsd.edu/data/amazon/}}: 
\begin{itemize}
\item \textbf{\textit{MovieLens-1M}} \cite{harper2016movielens}: This movie rating dataset has been
widely used to evaluate collaborative filtering algorithms. We used the version containing 1,000,209 ratings from 6,040 users for 3,900 movies. While it is a dataset with explicit feedbacks, we follow the convention \cite{neuralCF} that transforms it into implicit data, where each entry is marked as 0 or 1 indicating whether the user has rated the item. After transforming, we retain a dataset of $\mathbf{1.0\%}$ density.
\item \textit{\textbf{HetRec}} \cite{Cantador:RecSys2011}: This dataset has been released by the Second International Workshop on Information Heterogeneity and Fusion in Recommender Systems\footnote{http://ir.ii.uam.es/hetrec2011/}. It is an extension of \textit{MovieLens-10M} dataset and contains 855,598 ratings, 2,113 users and 10,197 movies. After converting it into implicit data as \textit{MovieLens-1M}, we obtain a dataset of $\mathbf{0.3\%}$ density.
\item \textbf{\textit{Amazon Instant Video}} \cite{mcauley2015image}: The dataset consists of 426,922 users, 23,965 videos and 583,933 ratings from \textit{Amazon.com}. Similarly, we transformed it into implicit data and removed users with less than 5 interactions. As a result, a dataset of \textbf{0.12\%} density is obtained.
\end{itemize}
\subsection{Experimental Setting}
Ideally, a recommendation model should not only be able to retrieve all relevant items out of all items but also provide a rank for each user where relevant items are expected to be ranked in the top. Therefore, in our experiments, we use Recall@M and MAP@M to evaluate the performance of the top-M recommendations. Recall@M is employed to measure the fraction of relevant items retrieved out of all relevant items. MAP@M is used for evaluating the ranking performance of RS. The Recall@M for each user is then defined as:
\begin{equation}
\text{Recall@M} = \frac{\text{\#items\; the\; user\; likes\; among\; the\; top\; M}}{\text{total\; number\; of\; items\; the\; user\; likes}}.
\end{equation}
The final results reported are average recall over all users. %The definition of MAP@M is given as:
%\begin{eqnarray}
%\text{MAP@M} = \frac{1}{|\mathcal{U}|}\sum_{r \in \mathcal{U}}AveP(r,M), \nonumber \\
%AveP(r,M) = \frac{1}{M} \sum_{m=1}^{M}P_r(m)rel_r(m),
%\end{eqnarray}
%where $P_r(m)$ indicates the precision at cut-off $m$ in the product list recommended to user $r$; $rel_r(m)$ is an indicator function denoting whether the $m_\textrm{th}$ item in the recommendation list is viewed/liked/clicked by user $r$.

For each dataset, we randomly select 80\% items associated with each user to constitute the training set and use all the remaining as the test set. For each evaluation scenario, we repeat the evaluation five times with different randomly selected training sets and the average performance is reported in the following sections. 

We use a validation set from the training set of each dataset to find the optimal hyper-parameters of comparative methods introduced in the Section \ref{sec:baseline}. For ItemKNN, we employ the cosine distance to measure item similarities. The dimensions of latent factors for BPR, eALS and GNMF are searched from \{8,16,32,64,128\} via the validation set. The hyperparameter $\lambda$ of eALS is selected from 0.001 to 0.04. Since the architecture of a multi-layer perceptron (MLP) is difficult to optimize, we follow the suggestion from the original paper \cite{neuralCF} to employ a three-layer MLP with the shape of $(32,16,8)$ for NCF. The dropout rate of nodes for GCMC is searched from \{0.3,0.4,0.5,0.6,0.7,0.8\}. Our SpectralCF has one essential hyper-parameter: $K$. Figure \ref{fig:：parameter} shows how the performances of SpectralCF vary as $K$ is set from 1 to 5 on the validation set of \textit{MovieLens-1M}. As we can see, in terms of Recall@20 and MAP@20, SpectralCF reaches its best performances when $K$ is fixed as 3. Other hyper-parameters of SpectralCF are empirically set and summarized in Table \ref{table::setting}, where $\lambda$ denotes the learning rate of RMSprop. %This setting achieves the best performances for SpectralCF in all three datasets. 
Our models are implemented in \textit{TensorFlow} \cite{abadi2016tensorflow}. %All models are trained and tested on an NVIDIA TITAN X GPU.
\subsection{Experimental Results (RQ1 and RQ2)}
\begin{figure*}[t]
\centering
\begin{center}
\begin{subfigure}{0.33\textwidth}
\centering
\includegraphics[height=.17\textheight]{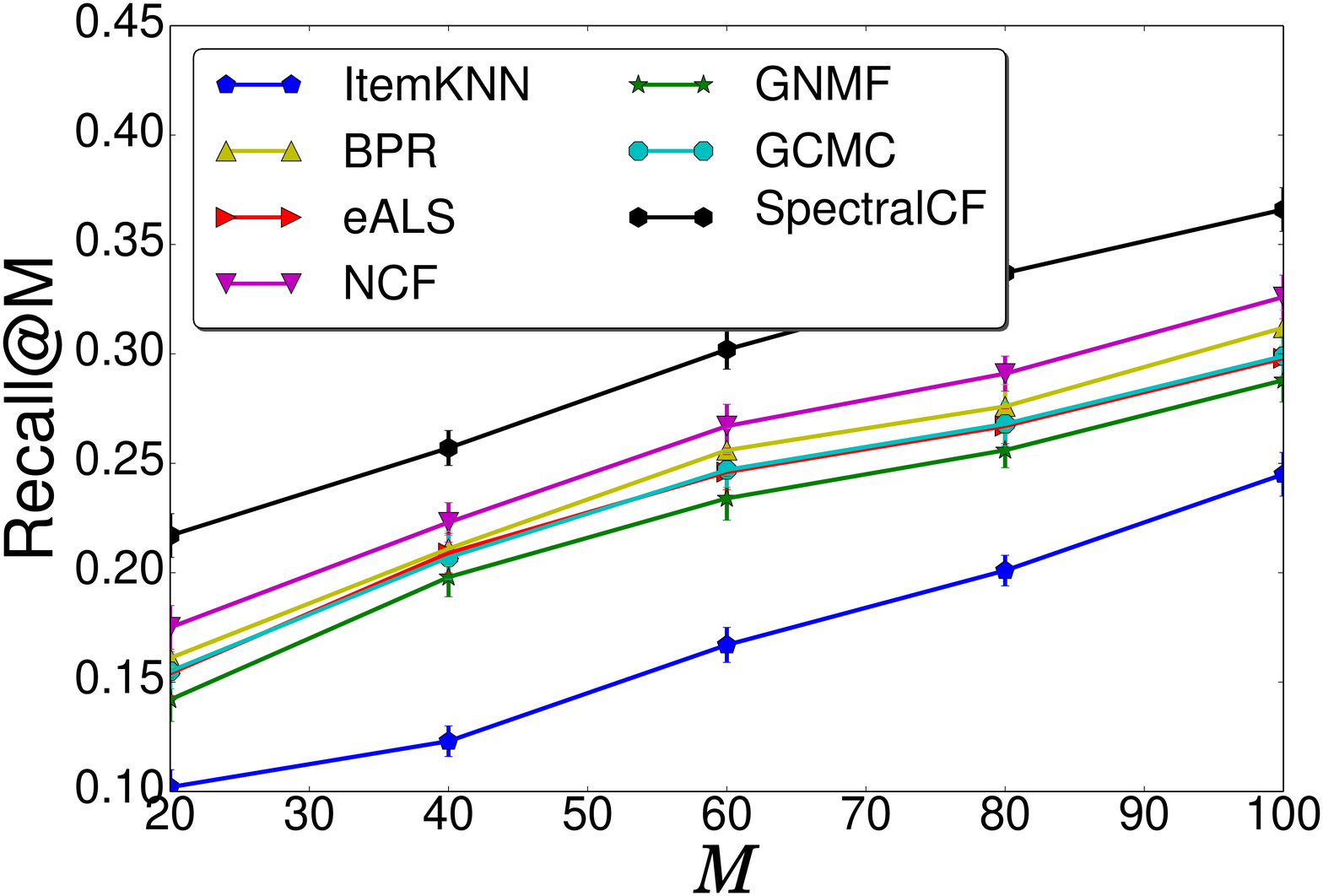}
\caption{MovieLens-1M}
\end{subfigure}
\begin{subfigure}{0.33\textwidth}
\centering
\includegraphics[height=.17\textheight]{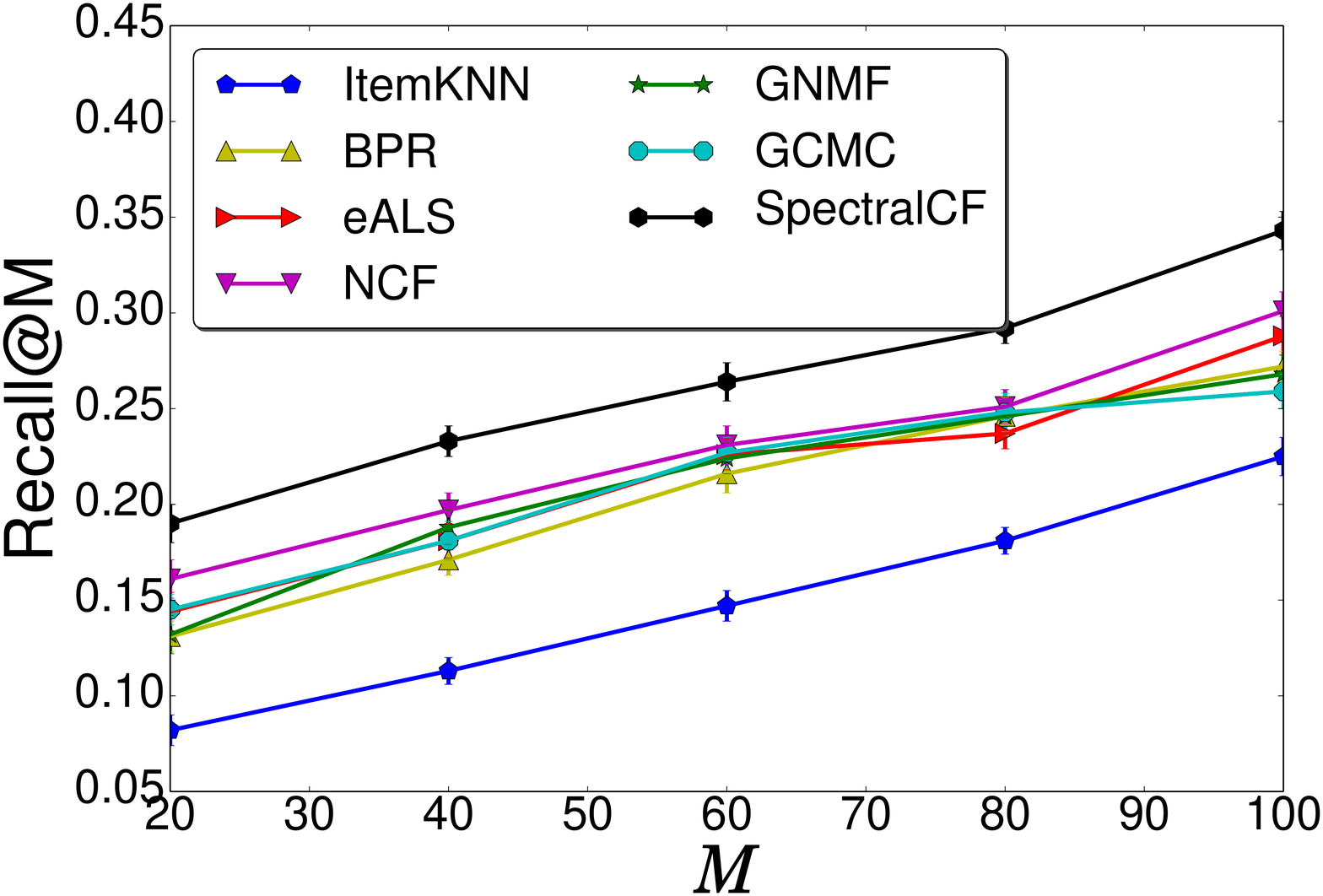}
\caption{HetRec}
\end{subfigure}
\begin{subfigure}{0.33\textwidth}
\centering
\includegraphics[height=.17\textheight]{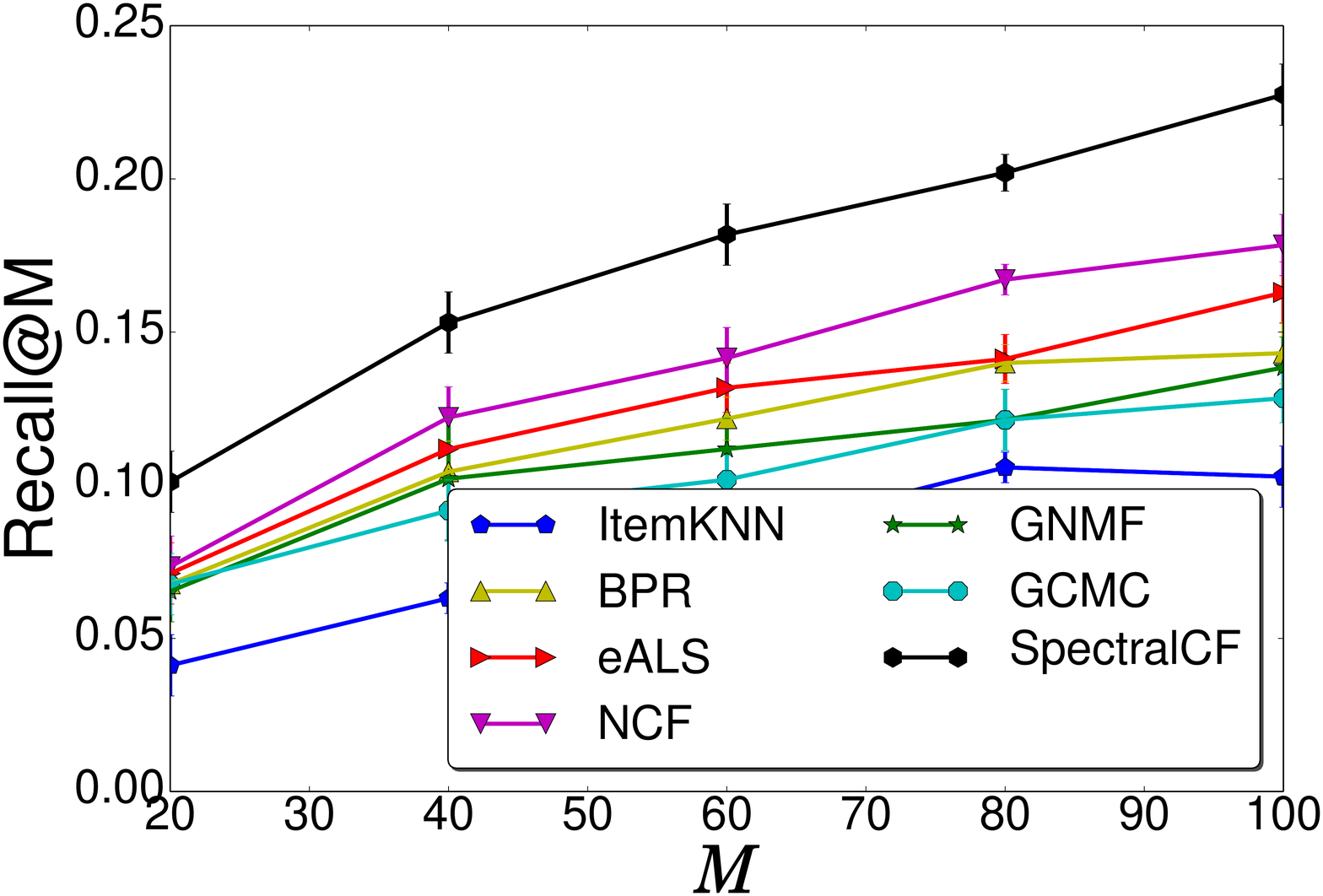}
\caption{Amazon Instant Video}
\end{subfigure}
\end{center}
\vspace{-1em}
\caption{Performance comparison in terms of recall@M with M varied from 20 to 100. Errors bars are 1-standard deviation.}
\label{exp::recall}
\end{figure*}
\begin{figure*}[t]
\centering
\begin{center}
\begin{subfigure}{0.33\textwidth}
\centering
\includegraphics[height=.17\textheight]{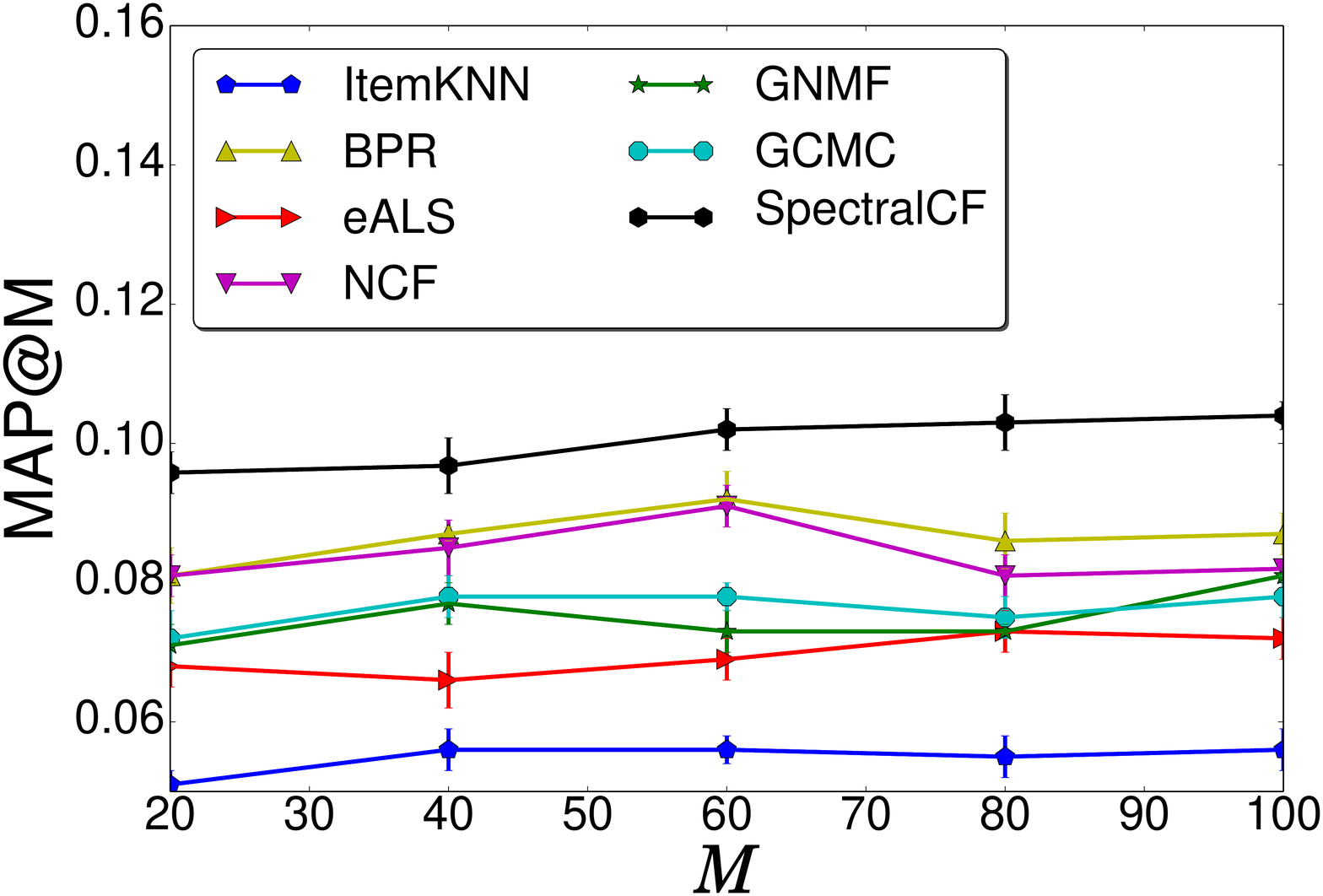}
\caption{MovieLens-1M}
\end{subfigure}
\begin{subfigure}{0.33\textwidth}
\centering
\includegraphics[height=.17\textheight]{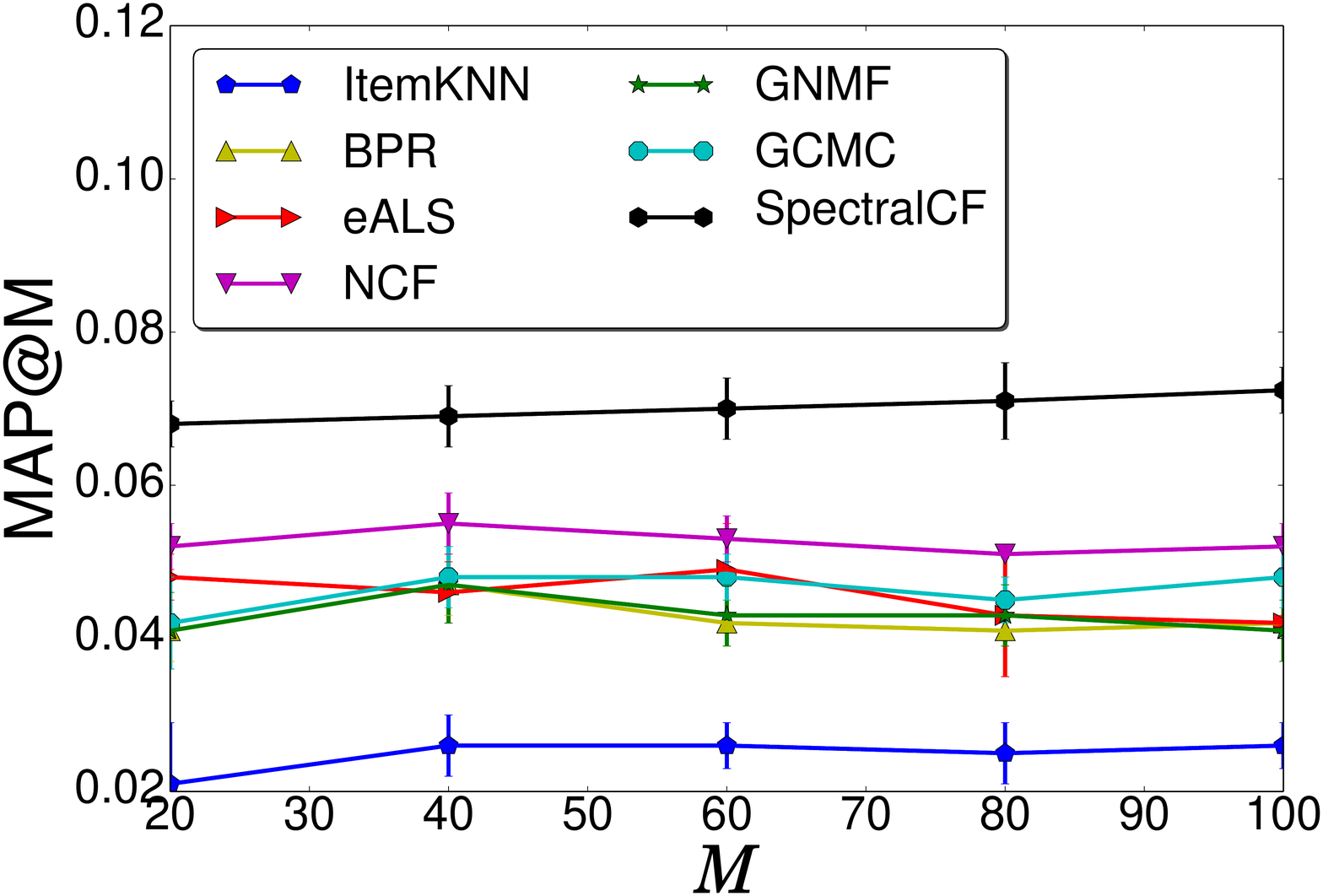}
\caption{HetRec}
\end{subfigure}
\begin{subfigure}{0.33\textwidth}
\centering
\includegraphics[height=.17\textheight]{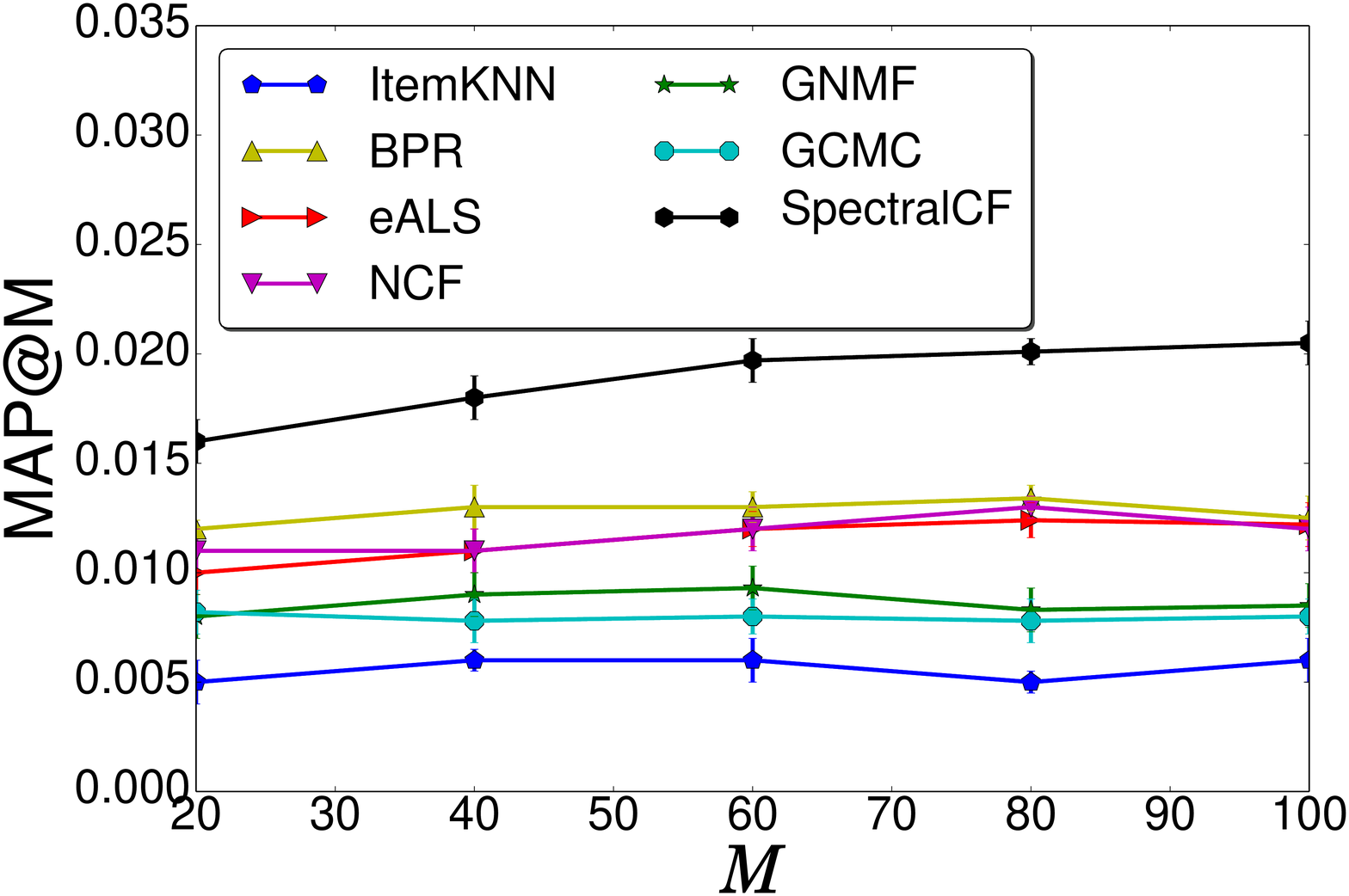}
\caption{Amazon Instant Video}
\end{subfigure}
\end{center}
\vspace{-1em}
\caption{Performance comparison in terms of MAP@M with M varied from 20 to 100. Errors bars are 1-standard deviation.}
\label{exp::map}
\end{figure*}
In Figure \ref{exp::recall}, we compare SpectralCF with four CF-based models and two graph-based models in terms of Recall@M on all three datasets. Overall, when $M$ is varied from 20 to 100, SpectralCF consistently yields the best performance across all cases. Among CF-based comparative models, ItemKNN gives the worst performances in all three datasets, indicating the necessity of modeling users' personalized preferences rather than just recommending similar items to users. For graph-based models (GNMF and GCMC), they generally underperform CF-based models such as BPR and NCF. The unsatisfying performance of GNMF shows that adding a graph-based regularization is not sufficient to capture complex structures of graphs. Though GCMC directly performs on a user-item bipartite graph, each vertex in the graph is only allowed to learn from its neighbors. This constrains its ability of capturing global structures in the graph. Among all comparative models, benefiting from its capability of modeling non-linear relationships between users and items, NCF beats all other models and becomes the strongest one. However, none of models above are able to directly perform in the \textit{spectral domain}. They lose the rich information in the domain and as a result, SpectralCF greatly outperforms NCF by \textbf{16.1\%}, \textbf{16.2\%} and \textbf{28.0\%} in the dataset of \textit{MovieLen-1M}, \textit{HetRec} and \textit{Amazon Instant Video}, respectively.

In Figure \ref{exp::map}, we compare SpectralCF with all comparative models in terms of MAP@M. Again, when $M$ is in a range from 20 to 100, SpectralCF always yields the best performance. Neighbor-based ItemKNN performs the worst among all models. It further shows the advantages of modeling users' personalized preferences. Compared with NCF and BPR, graph-based models (GNMF and GCMC) again fail to show convincing ranking performances measured by MAP@M. For CF-based models, while NCF beats other CF-based models in the dataset of \textit{HetRec}, BPR shows itself as a strong model for ranking, owing to its pairwise ranking loss. It slightly outperforms NCF on average in the datasets of \textit{MovieLens-1M} and \textit{Amazon Instant Video}. However, SpectralCF improves BPR by \textbf{15.9\%}, \textbf{64.9\%} and \textbf{47.5\%} in the dataset of \textit{MovieLen-1M}, \textit{HetRec} and \textit{Amazon Instant Video}, respectively. %NCF by \textbf{32.6\%}, \textbf{33.3\%} and \textbf{59.7\%} and

Overall, as shown in Figure~\ref{exp::recall} and~\ref{exp::map}, not surprisingly, the performances of all models decline as the dataset becomes sparse. However, SpectralCF always outperforms all comparative models regardless of the sparsities of the datasets. By comparing spectralCF with traditional CF-based models, we demonstrate that the rich information of connectivity existing in the \textit{spectral domain} assists SpectralCF in learning better latent factors of users and items. By comparing SpectralCF with graph-based models, we show that SpectralCF can effectively learn from the \textit{spectral domain}.
\subsection{Quality of Recommendations for Cold-start Users (RQ3)}
To answer \textbf{RQ3}, in this section, we conduct an experiment to investigate the quality of recommendations made by SpectralCF for \textit{cold-start} users. To this end, in the dataset of \textit{MovieLens-1M}, we build training sets with different degrees of sparsity by varying the number of items associated with each user, denoted as $P$, from one to five. All the remaining items associated with users are used as the test set. We compare SpectralCF with BPR, which is widely known and also shown as a strong ranking performer in Figure \ref{exp::map}. The test results are reported in the Table~\ref{table::sparse}. 

In Table~\ref{table::sparse}, it is shown that, suffering from the \textit{cold-start} problem, the performances of BPR and SpectralCF inevitably degrade. However, regardless of the number of items associated with users, SpectralCF consistently outperforms BPR in terms of Recall@20 and MAP@20. On average, SpectralCF improves BPR by \textbf{36.8\%} and \textbf{33.8\%} in Recall@20 and MAP@20, respectively. Hence, it is demonstrated that compared with BPR, spectralCF can better handle \textit{cold-start} users and provide more reliable recommendations.  
\section{Related Works}
\label{sec:related}
\begin{table}[t]
\centering
\caption{Performance Comparison in terms of Recall@20 and MAP@20 in the sparse training sets. In the dataset of \textit{MovieLens-1M}, we vary the number of items associated with each users, denoted as $P$, from 1 to 5. The average results are reported and the best results are in bold. The standard deviation is shown in parentheses.}
\label{table::sparse}
\begin{tabular}{|M{0.65cm}||M{1.23cm}|M{0.8cm}|M{0.8cm}|M{0.8cm}|M{0.8cm}|M{0.8cm}|}
\hline
\multirow{4}{*}{} &  P&  1&  2&  3&  4& 5 \\ \cline{2-7} 
                  &  BPR& 0.021 (0.003)& 0.029 (0.004)& 0.031 (0.003)& 0.034 (0.004)& 0.038 (0.003) \\ \cline{2-7} 
                  Recall @20&  SpectralCF&\textbf{0.031} (0.003)&\textbf{0.039} (0.003)&\textbf{0.042} (0.002)&\textbf{0.045} (0.003)&\textbf{0.051} (0.003)  \\ \cline{2-7} 
                  &  Improve- ment&47.6\%&34.5\%&35.5\%&32.4\%&34.2\%  \\ \hline
\multirow{4}{*}{}  
                  &  BPR& 0.014 (0.002)& 0.017 (0.002)& 0.021 (0.002)& 0.024 (0.003)& 0.027 (0.003) \\ \cline{2-7} 
                  MAP @20&  SpectralCF&\textbf{0.019} (0.002)&\textbf{0.024} (0.002)&\textbf{0.028} (0.003)&\textbf{0.031} (0.003)&\textbf{0.035} (0.002)  \\ \cline{2-7} 
                  &  Improve- ment&35.7\%&41.2\%&33.3\%&29.2\%&29.6\% \\ \hline
\end{tabular}
\end{table}
There are two categories of studies related to our work: deep learning based RS and graph-based RS. In this section, we will first briefly review existing works in the area of deep RS. Then, we focus on presenting recent works on graph-based RS. Despite all these approaches, SpectralCF is the first model to directly learn latent factors of users and items from the \textit{spectral domains} of user-item bipartite graphs.
\subsection{Deep Recommender Systems}
One of the early works utilizing deep learning for RS builds a Restricted Boltzmann Machines (RBM) based method to model users using their rating preferences \cite{salakhutdinov2007restricted}. Although the method is still a relatively shallow model, it slightly outperforms Matrix Factorization technique and shows the promising future for deep recommender systems. In \cite{wang2017irgan}, a generative model and a discriminative model are employed to play a minimax game. The two models are iteratively optimized and achieve promising results for the item recommendation problem. Inspired by \cite{salakhutdinov2007restricted}, \cite{zheng2016neural} proposed a CF Neural Autoregressive Distribution Estimator (CF-NADE) model for collaborative filtering tasks. CF-NADE shares parameters between different ratings. \cite{neuralCF} presents to utilize a Multilayer Perceptron (MLP) to model user-item interactions. %\cite{hidasi2015session} proposes an RNN-based approach for session-based recommendations. They modify classic RNNs to make it more viable for RS. 

A number of researchers proposed to build a hybrid recommender systems to counter the sparsity problem. \cite{wang2014improving} introduce Convolutional Neural Networks (CNN) and Deep Belief Network (DBN) to assist representation learning for music data. As such, their model is able to extract latent factors of songs without ratings while CF based techniques like MF are unable to handle these songs. These approaches above pre-train embeddings of users and items with matrix factorization and utilize deep models to fine-tune the learned item features based on item content. %In \cite{van2013deep}, they initially find user and item latent factors using matrix factorization techniques. Then, a deep Convolutional Neural Networks (CNN) is established to reconstruct item latent factors from the music data. As such, their model is able to extract latent factors of songs without ratings while CF based techniques like MF are unable to handle these songs. Similar to \cite{van2013deep}, \cite{wang2014improving} also introduce Deep Belief Network (DBN) to assist representation learning for music data. In fact, the approaches above model users and items by a matrix factorization technique and utilize deep models to fine-tune the learned item features based on item information. 
%In \cite{elkahky2015multi}, a multi-view deep model is built to jointly estimate the user and item latent factors in a joint manner. The authors use one neural network for users' query histories. Thus, the neural network is referred as user view. And another neural network is responsible for learning from implicit feedbacks of items (e.g., News clicks, App downloads). The two networks are combined in a manner similar to Matrix Factorization. The resulting model is named multi-view DNN since it can incorporate item information from more than one domains and jointly optimize all of them using the same loss function. 
In \cite{elkahky2015multi} and \cite{wang2017deep}, multi-view deep models are built to utilize item information from more than one domain. 
\cite{kim2016convolutional} integrates a CNN with PMF to analyze documents associated with items to predict users' future explicit ratings. \cite{zheng2017joint} leverage two parallel neural networks to jointly model latent factors of users and items. To incorporate visual signals into RS, \cite{wang2017your} propose CNN-based models to incorporate visual signals into RS. They make use of visual features extracted from product images using deep networks to enhance the performance of RS. \cite{zhang2016collaborative} investigates how to leverage the multi-view information to improve the quality of recommender systems. \cite{cheng2016wide} jointly trains wide linear models and deep neural networks for video recommendations. \cite{wang2016collaborative} and \cite{zheng2017hierarchical} utilize RNN to consider word orders and extract complex semantics for recommendations. \cite{wang2017dynamic} applies an attention mechanism on a sequence of models to adaptively capture the change of criteria of editors. \cite{zheng2018mars} leverages an attentional model to learn adaptive user embeddings. A survey on the deep learning based RS with more works on this topic can be found in \cite{zhang2017deep}.

\subsection{Graph-based Recommender Systems}
In order to learn latent factors of users and items from graphs, a number of researchers have proposed graph-based RS. \cite{zhou2008learning} develops a semi-supervised learning model on graphs for document recommendation. The model combines multiple graphs in order to measure item similarities. In \cite{yuan2014graph}, they propose to model the check-in behaviors of users and a graph-based preference propagation
algorithm for point of interest recommendation. The proposed solution exploits both the geographical and temporal influences
in an integrated manner. \cite{guan2009personalized} addresses the problem of personalized tag recommendation by modeling it as a "query and ranking" problem.
Inspired by the recent success of graph/node embedding methods, 
\cite{berg2017graph} proposes a graph convolution network based model for recommendations. In \cite{berg2017graph}, a graph auto-encoder learns the structural information of a graph for latent factors of users and items. \cite{cai2008non} adds graph-based regularizations into the matrix factorization model to learn graph structures. Graph-regularized methods are developed for the problemm of matrix completion in \cite{rao2015collaborative}. \cite{monti2017geometric} combines a convolutional neural network and a recurrent neural network to model
the dynamic rating generation process. Although this work also considers the \textit{spectral domain}, they learn from a graph constructed from side information, such as genres or actors for movies. In contrast, our method learns directly from user-item bipartite graphs and does not require the side information. Thus, this work is not comparable to our method.

Additionally, some scholars have proposed to incorporate the heterogeneous information on a graph for recommendations. \cite{jamali2013heteromf} suggests a general latent factor model for entities in a graph. \cite{yu2013recommendation} introduces a recommendation model for implicit data by taking advantage of different item similarity semantics in the graph. \cite{shi2015semantic} introduces a semantic path based personalized recommendation method to predict the rating scores of users on items.

However, all works above are different from ours because they fail to consider the rich information in the \textit{spectral domains} of user-item bipartite graphs. Also, our study focuses on learning from the implicit feedbacks, and leaves incorporating the heterogeneous information in a graph and the item content for future works. 
\section{Conclusions}
\label{sec:con}
It is shown that the rich information of connectivity existing in the \textit{spectral domain} of a bipartite graph is helpful for discovering deep connections between users and items. In this paper, we introduce a new spectral convolution operation to directly learn latent factors of users and items from the \textit{spectral domain}. Furthermore, with the proposed operation, we build a deep feed-forward neural network based recommendation model, named Spectral Collaborative Filtering (SpectralCF). Due to the rich information of connectivity existing in the \textit{spectral domain}, compared with previous works, SpectralCF is capable of discovering deep connections between users and items and therefore, alleviates the \textit{cold-start} problem for CF. To the best of our knowledge, SpectralCF is the first CF-based method directly learning from the \textit{spectral domains} of user-item bipartite graphs. We believe that it shows the potential
of conducting CF in the \textit{spectral domain}, and will encourage future works in this direction.

In comparison with four state-of-the-art CF-based and two graph-based models, SpectralCF achieved $\textbf{20.1\%}$ and $\textbf{42.6\%}$ improvements averaging on three standard datasets in terms of Recall@M and MAP@M, respectively. 

Additionally, in the experiments, by varying the number of items associated with each user from 1 to 5, we build training sets with different degrees of sparsity to investigate the quality of recommendations made by SpectralCF for \textit{cold-start} users. By comparing SpectralCF with BPR, on average, SpectralCF improves BPR by \textbf{36.8\%} and \textbf{33.8\%} in Recall@20 and MAP@20, respectively. It is validated that SpectralCF can effectively ameliorate the \textit{cold-start} problem.
\begin{acks}
This work is supported in part by NSF through grants IIS-1526499, IIS-1763325, and CNS-1626432, and NSFC 61672313. This work is also partially supported by NSF through grant IIS-1763365 and by FSU through the startup package and FYAP award.
\end{acks}